\documentclass{article}
\usepackage[latin1]{inputenc}
\usepackage[T1]{fontenc}
\usepackage{amssymb,amscd}
\usepackage{amsfonts}
\usepackage{amsmath}
\usepackage{xspace}
\usepackage{enumerate}
\usepackage{stmaryrd}
\usepackage{color}
\usepackage{graphicx}
\usepackage{subfigure}
\newtheorem{proposition}{Proposition}[section]
\newtheorem{remark}{Remark}
\newtheorem{theorem}{Theorem}[section]
\newtheorem{lemma}[theorem]{Lemma}
\newtheorem{corollary}[theorem]{Corollary}
\newtheorem{definition}[theorem]{Definition}

\newcommand{\qed}{\hfill$\square$}
\newenvironment{proof}{\par\noindent\textit{Proof.~}}{\qed\par\bigskip}
\newenvironment{keywords}{\bigskip\noindent\textbf{Keywords:} }{\par}
\newcounter{exnum}
\newenvironment{example}{\refstepcounter{exnum}%
                          \begin{trivlist}
                            \item[\hskip\labelsep{\bf Example
                                                      \arabic{exnum}}]}%
                         {\hfill\hbox{$\quad\Box$}\end{trivlist}}
\newcommand{\cF}{{\mathcal{F}}}
\newcommand{\ccL}{{\mathcal{L}}}
\newcommand{\z}{\mathbb{Z}}
\newcommand{\Z}{\z}
\newcommand{\Zt}{\widetilde{\z}}
\newcommand{\n}{\mathbb{N}}
\newcommand{\re}{{\mathbb R}}
\newcommand{\N}{\n}
\newcommand{\az}{\ensuremath{A^{\Z^d}}}
\newcommand{\C}{\ensuremath{\mathcal{C}}\xspace}
\newcommand{\B}{\ensuremath{\mathcal{B}}\xspace}

\newcommand{\cst}[1]{\ensuremath{\underline{#1}}}
\newcommand{\obrx}[1]{\beta_r^0(x_{#1})}

\newcommand{\obri}[1]{\beta_r^{x_i}(x_{#1})}
\newcommand{\obrj}[1]{\beta_r^{x_j}(x_{#1})}

\newcommand{\obmx}[1]{\beta_m^0(x_{#1})}

\newcommand{\oblFx}[1]{\beta_l^0(F(x)_{#1})}
\newcommand{\obry}[1]{\beta_r^0(y_{#1})}

\newcommand{\obrjy}[1]{\beta_r^{y_j}(y_{#1})}

\newcommand{\obmy}[1]{\beta_m^0(y_{#1})}

\newcommand{\oblFy}[1]{\beta_l^0(F(y)_{#1})}
\newcommand{\sa}{\ensuremath{{\mathcal{A}}}\xspace}
\renewcommand{\S}{\ensuremath{{\mathcal{S}}}\xspace}
\newcommand{\A}{\sa}
\newcommand{\SN}{\ensuremath{{\mathcal{N}}}\xspace}
\newcommand{\fSN}{\ensuremath{f_{\mathcal{N}}}\xspace}
\newcommand{\FSN}{\ensuremath{F_{\mathcal{N}}}\xspace}
\newcommand{\Rrd}{\ensuremath{\mathcal{R}^d_r}\xspace}
\newcommand{\Md}{\ensuremath{\mathcal{M}^d}\xspace}
\newcommand{\Mdd}{\ensuremath{\mathcal{M}^{d+1}}\xspace}
\newcommand{\Mrd}{\ensuremath{\mathcal{M}^d_{\mathbf{2r+1}}}\xspace}
\newcommand{\Mr}{\ensuremath{M_{\mathbf{2r+1}}}\xspace}
\newcommand{\rr}{{\mathbf r}}
\newcommand{\uno}{{\mathbf 1}}
\newcommand{\Mhd}{\ensuremath{\mathcal{M}^d_h}\xspace}
\newcommand{\ball}[1]{\left[#1\right]}

\renewcommand{\llbracket}{[} % commenter pour avoir des crochets doubles
\renewcommand{\rrbracket}{]} %
\newcommand{\inter}[2]{\left\llbracket #1,#2 \right\rrbracket}
\newcommand{\intinf}[2]{\widetilde{\inter{#1}{#2}}}
\newcommand{\dist}{\mathsf{d}}
\newcommand{\oldd}{\dist^\prime}
\newcommand{\oldc}{{C^{\prime}}}
\newcommand{\uple}[1]{\left\langle #1 \right\rangle}
\newcommand{\set}[1]{\left\{#1\right\}}
\newcommand{\parto}[1]{\left ( #1 \right )}
\newcommand{\para}[1]{\left\langle#1\right\rangle}
\newcommand{\nv}{0}
\newcommand{\Ml}{\ensuremath{M_{\mathbf{2l+1}}}\xspace}
\newcommand{\lu}{{\mathbf l}}

\newcommand{\ignore}[1]{}
\newcommand{\ie}{i.e.\xspace}
\newcommand{\TODO}[2][]{\textcolor{red}{TODO #1 $\Big\langle$}#2\textcolor{red}{$\Big\rangle$}}
\newcommand{\NEW}[1]{#1}
\newcommand{\NEWa}[1]{#1}

\title{A compact topology for Sand Automata\protect\thanks{This work has been supported
by the Interlink/MIUR project ``Cellular Automata: Topological
Properties, Chaos and Associated Formal Languages'', by the ANR
Blanc ``Projet Sycomore'' and by the PRIN/MIUR project ``Formal
Languages and Automata: Mathematical and Applicative Aspects''.}~~\protect\thanks{Some of the results of this paper have been submitted at JAC 2008 and IFIP TCS 2008 conferences.}}
\author{Alberto Dennunzio\footnote{Universit\`a degli studi di Milano-Bicocca,
Dipartimento di Informatica Sistemistica e Comunicazione,  via
Bicocca degli Arcimboldi 8, 20126 Milano (Italy).\protect\\Email: \texttt{dennunzio@disco.unimib.it}} \and Pierre Guillon\footnote{Universit\'e Paris-Est, Laboratoire d'Informatique de l'Institut Gaspard Monge, UMR CNRS 8049, 5 bd Descartes,
 77 454 Marne la Vall\'ee Cedex 2 (France).\protect\\Email: \texttt{pierre.guillon@univ-mlv.fr}} \and Beno\^it Masson\footnote{Laboratoire d'Informatique Fondamentale de Marseille
 (LIF)-CNRS, Aix-Marseille Universit\'e,
 39 rue Joliot-Curie, 13 453 Marseille Cedex 13
 (France).\protect\\Email: \texttt{benoit.masson@lif.univ-mrs.fr}}}
\date{}

\begin{document}

\maketitle

\begin{abstract}
  In this paper, we exhibit a strong relation between the sand
  automata configuration space and the cellular automata configuration
  space. \NEWa{This relation induces a compact topology
    for sand automata, and a new context in which sand automata are
    homeomorphic to cellular automata acting on a specific subshift.
    We show that the existing topological results for sand automata,
    including the Hedlund-like representation theorem, still hold.  In
    this context, we give a characterization of the cellular automata
    which are sand automata, and study some dynamical behaviors such as
    equicontinuity.}
  Furthermore, we deal with the nilpotency. We show that the classical
  definition is not meaningful for sand automata. Then, we introduce a
  suitable new notion of nilpotency for sand automata. Finally, we
  prove that this simple dynamical behavior is undecidable.
\end{abstract}
%----
%\subjclass{F.1.1 Models of Computation}
%---
%---

\section{Introduction}
Self-organized criticality (SOC) is a common phenomenon observed
in a huge variety of processes in physics, biology and computer
science. A SOC system evolves to a ``critical state'' after some
finite transient. Any perturbation, no matter how small, of the
critical state generates a deep reorganization of the whole
system. Then, after some other finite transient, the system
reaches a new critical state and so on. Examples of SOC systems
are: sandpiles, snow avalanches, star clusters in the outer space,
earthquakes, forest fires, load balance in operating systems
\cite{bak97,bak88,bak89,bak90,subramanian94}. Among them,
sandpiles models are a paradigmatic formal model for SOC systems
\cite{goles93,GMP}.

In \cite{CF03}, the authors introduced sand automata as a
generalization of sandpiles models and transposed them in the
setting of discrete dynamical systems. A key-point of \cite{CF03}
was to introduce a (locally compact) metric topology to study the
dynamical behavior of sand automata. A first and important result
was a fundamental representation theorem similar to the well-known
Hedlund's theorem for cellular automata~\cite{hedlund69,CF03}.
In~\cite{cervelle05,CFM07}, the authors investigate sand automata
by dealing with some basic set properties and decidability issues.
%They show that many relations
%between set properties that are true for cellular automata turned
%out/are no more true for sand automata. This fact leads to
%conclude that sand automata are a new model and not a peculiar
%sub-model of cellular automata.

In this paper we continue the study of sand automata. First of all, we
introduce a different metric on configurations (\ie spatial
distributions of sand grains). This metric is defined by means of the
relation between sand automata and cellular automata~\cite{CFM07}.
With the induced topology, the configuration set turns out to be a
compact (and not only locally compact), perfect and totally
disconnected space. The ``strict'' compactness gives a better
topological background to study the behavior of sand automata (and in
general of discrete dynamical systems). In fact, compactness provides
a lots of very useful results which help in the investigation of
several dynamical properties~\cite{ak93,ku04}. We show that all the
topological results from~\cite{CF03} still hold, in particular the
Hedlund-like representation theorem remains valid with the compact
topology. \NEWa{Moreover, with this topology, any sand automaton is
  homeomorphic to a cellular automaton defined on a subset of its
  usual domain. We prove that it is possible to decide whether a given
  cellular automaton is in fact a sand automaton. Besides, this
  relation helps to prove some properties about the dynamical behavior
  of sand automata, such as the equivalence between equicontinuity and
  ultimate periodicity.}

Then, we study nilpotency of sand automata. The classical
definition of nilpotency for cellular automata~\cite{CPY89,K92} is
not meaningful, since it prevents any sand automaton from being
nilpotent. Therefore, we introduce a new definition which captures
the intuitive idea that a nilpotent automaton destroys all the
configurations: a sand automaton is nilpotent if all
configurations get closer and closer to a uniform configuration,
not necessarily reaching it. Finally, we prove that this behavior
is undecidable.

\NEWa{The paper is structured as follows. First, in Section~\ref{sec:def},
we recall basic definitions and results about cellular automata and
sand automata. Then, in Section~\ref{sec:topo}, we define a compact
topology and we prove some topological results, in particular the
representation theorem. Finally, in Section~\ref{sec:nil}, nilpotency
for sand automata is defined and proved undecidable.}

\ignore{
% phrases utiles pour l'importance de la compacite'.
\begin{enumerate}
\item X will be a compact metric space most of the time. Note that
  Baire Category Theorem has life on compact metric spaces. This
  Theorem is invariably invoked to get results of the following form:
  the set of points satisfying such and such dynamical properties is a
  dense $G_\delta$.
\item Baire Theorem...
\item we began the study of aleph-Cellular Automata, \ie CA whose
  cells are defined on an infinite set of states. The study of these
  DTDS is complicated by the fact that the state space is not compact.
\end{enumerate}
}
\section{Definitions}\label{sec:def}

For all $a,b\in\Z$ with $a\leq b$, let
$\inter{a}{b}=\set{a,a+1,\ldots,b}$ and
$\intinf{a}{b}=\inter{a}{b}\cup\set{+\infty,-\infty}$. For $a \in \Z$,
let $[a, +\infty) = \set{a, a+1, \ldots} \setminus \set{+\infty}$.
Let $\n_+$ be the set of positive integers. For a vector $i\in\Z^d$,
denote by $|i|$ the infinite norm of $i$.

\NEWa{Let $A$ a (possibly infinite) alphabet and $d\in\N^*$. Denote by $\Md$ the set of all the $d$-dimensional matrices with
values in $A$. We assume that the
entries of any matrix $U\in\Md$ are all the integer vectors of a
suitable $d$-dimensional hyper-rectangle $[1,h_1]\times\cdots\times
[1,h_d]\subset\n_+^d$. For any $h=(h_1,\ldots,h_d)\in\n_+^d$, let
$\Mhd\subset \Md$ be the set of all the matrices with entries in
$[1,h_1]\times\cdots\times [1,h_d]$.  In the sequel, the vector $h$
will be called the \emph{order} of the matrices belonging to $\Mhd$. For a
given element $x\in\az$, the \emph{finite portion} of $x$ of
reference position $i\in\Z^d$ and order $h\in\n_+^d$ is the matrix
$M^i_h(x)\in\Mhd$ defined as $\forall k\in [1,h_1]\times\cdots\times
[1,h_d]$, $M^i_h(x)_k=x_{i+k-\uno}$.
For any $r\in\n$, let $\rr^d$ (or simply $\rr$ if the dimension is not ambiguous) be the vector $(r,\ldots,r)$.
\ignore{ Using this notation, we
  define for any $r\in\n$ the \emph{neighborhood} of $x$ with
  center $i$ and radius $r$, as the matrix $\Mr^{i-\rr}(x)\in\Mrd$.}}

\subsection{Cellular automata and subshifts}
%
%In this section we recall some useful definitions and results
%about CA {\bf and...}.\smallskip\\
Let $A$ be a finite alphabet.  A \emph{CA configuration} of
dimension $d$ is a function from $\z^d$ to $A$. The set $\az$ of
all the CA configurations is called the \emph{CA configuration
space}.
%The definition of \emph{factor} can be naturally
%extended to configurations: for $x\in\az$ and
%$i\le j$, $x\subi ij\defeq x_i\ldots x_j\fact x$.
%If $u\in A^*$, then $u^\omega$ is the infinite word
%of $\an$ consisting in periodic repetitions
%of $u$ and $\pinf u^\omega$ is the configuration of $\az$
%consisting in periodic repetitions of $u$.
This space is usually equipped with the Tychonoff metric $\dist_T$
defined by
\[
  \forall x,y\in\az,\quad\dist_T(x,y)=2^{-k}\quad\text{where}\quad k=\min\set{|j|: j\in \Z^d, x_j\ne y_j}\enspace.
\]
The topology induced by $\dist_T$ coincides with the product topology
induced by the discrete topology on $A$.  \NEWa{With this topology,
  the CA configuration space is a Cantor space: it is compact,
  perfect (\ie, it has no isolated points) and totally disconnected.}

\smallskip

\NEW{For any $k\in\Z^d$ the \emph{shift map} $\sigma^k: A^{\Z^d} \to
  A^{\Z^d}$ is defined by $\forall x\in A^{\Z^d}, \forall i\in\Z^d$,
  $\sigma^k(x)_i=x_{i+k}$. A function $F:A^{\Z^d} \to A^{\Z^d}$ is
  said to be \emph{shift-commuting} if $\forall k\in\Z^d$,
  $F\circ\sigma^k=\sigma^k\circ F$.}

A $d$-dimensional \emph{subshift} $S$ is a closed subset of the CA
configuration space $A^{\Z^d}$ which is shift-invariant,
\ie for any $k \in \Z^d$, $\sigma^k(S)
\subset S$. Let $\cF\subseteq\Md$ and let $S_{\cF}$ be the set of
configurations $x\in A^{\Z^d}$ such that all possible finite portions
of $x$ do not belong to $\cF$, \ie for any $i,h\in\Z^d$,
$M^i_h(x)\notin\cF$. The set $S_{\cF}$ is a subshift, and $\cF$ is
called its set of forbidden patterns. \NEWa{Note that for any subshift
  $S$, it is possible to find a set of forbidden patterns $\cF$ such
  that $S=S_\cF$}. A subshift $S$ is said to be a \emph{subshift of
  finite type} (SFT) if $S=S_{\cF}$ for some finite set $\cF$. The
language of a subshift $S$ is $\ccL(S)=\set{U\in\Md: \exists
  i\in\Z^d,h\in\N_+^d,x\in S,M_h^i(x)=U}$ (for more on subshifts,
see~\cite{LiMa95} for instance).

\smallskip

A \emph{cellular automaton} is a quadruple $\uple{A, d, r, g}$, where
$A$ is the alphabet also called the \emph{state set}, $d$ is the
dimension, $r\in\n$ is the \emph{radius} and $g: \Mrd \to A$ is the
\emph{local rule} of the automaton. The local rule $g$ induces a
\emph{global rule} $G:\az\to\az$ defined as follows,
  \[
  \forall x\in \az,\,\forall i\in\Z^d,\quad G(x)_i=
    g\big(\Mr^{i-\rr}(x)\big)\enspace .
  \]
Note that CA are exactly the class of all shift-commuting functions
which are (uniformly) continuous with respect to the Tychonoff
metric (Hedlund's theorem from~\cite{hedlund69}). For the sake of
simplicity, we will make no distinction between a CA and its
global rule $G$.

The local rule $g$ can be extended naturally to all finite
matrices in the following way. With a little abuse of notation,
for any $h\in[2r+1,+\infty)^d$ and any $U\in\Mhd$, define $g(U)$ as the
matrix obtained by the simultaneous application of $g$ to all the
$\Mrd$ submatrices of $U$. Formally, $g(U)=M^\rr_{h-2\rr}(G(x))$,
where $x$ is any configuration such that $M^0_h(x)=U$.

For a given CA, a state $s\in A$ is \emph{quiescent} (resp.,
\emph{spreading}) if for all matrices $U\in\Mrd$ such that $\forall
k\in [1,2r+1]^d$, (resp., $\exists k\in [1,2r+1]^d$) $U_k=s$, it holds
that $g(U) = s$.  Remark that a spreading state is also quiescent. A
CA is said to be spreading if it has a spreading state. In the sequel,
we will assume that for every spreading CA the spreading state is
$0\in A$.

\subsection{SA Configurations}

% Denote $\N = \set{0, 1, \ldots}$ the set of positive integers and $\N^* = \set{1, 2, \ldots}$ the set of non-zero positive integers.

A \emph{SA configuration} (or simply \emph{configuration}) is a set of
sand grains organized in piles and distributed all over the
$d$-dimensional lattice $\z^d$. A \emph{pile} is represented either by
an integer from $\Z$ (\emph{number of grains}), or by the value
$+\infty$ (\emph{source of grains}), or by the value $-\infty$
(\emph{sink of grains}), \ie it is an element of $\Zt=\Z\cup\{-\infty,
+\infty\}$. One pile is positioned in each point of the lattice
$\z^d$. Formally, a configuration $x$ is a function from $\Z^d$ to
$\Zt$ which associates any vector $i=(i_1, \ldots, i_d) \in \Z^d$ with
the number $x_i\in\Zt$ of grains in the pile of position $i$. When the
dimension $d$ is known without ambiguity we note $0$ the null vector
of $\Z^d$. Denote by $\C=\Zt^{\Z^d}$ the set of all configurations.
%When the dimension $d$ is known
%without ambiguity we may write $0$ for the vector $(0,\ldots,0)$
%of $\Z^d$.
%Denote by $|i|$ the infinite norm of vector $i \in Z^d$.
A configuration $x\in \C$ is said to be \emph{constant} if there
is an integer $c \in \Z$ such that for any vector $i \in \Z^d$,
$x_i = c$. In that case we write $x = \cst{c}$. A configuration
$x\in\C$ is said to be \emph{bounded} if there exist two integers
$m_1, m_2 \in \Z$ such that for all vectors $i \in \Z^d$, $m_1 \leq x_i
\leq m_2$. Denote by \B the set of all bounded configurations.
\smallskip

A \emph{measuring device} $\beta_r^m$ of precision $r\in\N$ and
reference height $m\in\Z$ is a function from $\Zt$ to $\intinf{-r}{r}$
defined as follows
\begin{equation*}\label{measdev}
  \forall n\in\Zt, \quad \beta_r^m(n)=
    \left\{\begin{array}{l@{\hspace{5mm}}l}
      +\infty &\text{if $n > m + r$}\enspace,\\
      -\infty &\text{if $n < m - r$}\enspace,\\
      n - m   & \text{otherwise.}
    \end{array}\right.
\end{equation*}
A measuring device is used to evaluate the relative height of two
piles, with a bounded precision. This is the technical basis of
the definition of cylinders, distances and ranges which are used
all along this article.
\smallskip
%---

In \cite{CF03}, the authors equipped $\C$ with a metric in such a way
that two configurations are at small distance if they have the same
number of grains in a finite neighborhood of the pile indexed by the
null vector. The neighborhood is individuated by putting the measuring
device %from Equation \eqref{measdev}
at the top of the pile, if this latter contains a finite number of
grains. Otherwise the measuring device is put at height $0$. In order
to formalize this distance, the authors introduced the notion of
\emph{cylinder}, that we rename \emph{top cylinder}. For any
configuration $x\in \C$, for any $r\in\n$, and for any $i\in\z^d$, the
top cylinder of $x$ centered in $i$ and of radius $r$ is the
$d$-dimensional matrix \NEWa{$\oldc^i_r(x) \in \Mrd$ defined on the
  infinite alphabet $A=\Zt$ by}
\[
  \forall k \in \inter{1}{2r+1}^d, \; \left(\oldc^i_r(x)\right)_k =
  \left\{\begin{array}{l@{\hspace{3mm}}l@{}}
       x_i &\text{if $k=r+1$}\enspace,\\
      \obri{i+k-r-1} &\text{if $k\neq r+1$ and $x_i\neq\pm\infty$}\enspace,\\
      \obrx{i+k-r-1}  & \text{otherwise.}
    \end{array}\right.
\]
In dimension $1$ and for a configuration $x\in\C$, we have
\[  \oldc^i_r(x) = \parto{\obri{i-r}, \ldots, \obri{i-1},
  x_{i}, \obri{i+1}, \ldots, \obri{i+r}}
\]
%---
if $x_i\neq\pm\infty$, while
%---
\[
  \oldc^i_r(x) = \parto{\obrx{i-r}, \ldots, \obrx{i-1}, x_{i},
  \obrx{i+1}, \ldots, \obrx{i+r}}
\]
%---
if $x_i=\pm\infty$.
\smallskip\\
%---
By means of top cylinders, the distance $\oldd:\C\times\C\to
\re_+$ has been introduced as follows:
\begin{equation*}
\forall x,y\in\C, \quad \oldd(x,y)=2^{-k} \quad \text{where}\quad
k=\min\set{r\in\n: \oldc^0_r(x)\neq \oldc^0_r(y)} \enspace.
\end{equation*}
%---
\begin{proposition}[\cite{CF03,CFM07}]
  With the topology induced by $\oldd$, the configuration space is
  locally compact, perfect and totally disconnected.
\end{proposition}
%---

\subsection{Sand automata}
For any integer $r \in \N$, for any configuration $x \in \C$ and
any index $i \in \Z^d$ with $x_i\neq\pm\infty$, the \emph{range} of center $i$ and radius $r$
is the $d$-dimensional matrix \NEWa{$R^i_r(x) \in \Mrd$ on the finite
  alphabet $A = \intinf{-r}{r} \cup \bot$ such that}
\[
  \forall k \in \inter{1}{2r+1}^d, \quad \left(R^i_r(x)\right)_k =
  \left\{\begin{array}{l@{\hspace{5mm}}l}
    \bot & \text{if $k=r+1$}\enspace, \\
    \beta^{x_i}_r(x_{i+k-r-1}) & \text{otherwise.}
  \end{array}\right.
\]
The range is used to define a sand automaton. It is a kind of top
cylinder, where the observer is always located on the top of the
pile $x_i$ (called the \emph{reference}). It represents what the
automaton is able to see at position $i$. Sometimes the central
$\bot$ symbol may be omitted for simplicity sake. The set of all
possible ranges of radius $r$, in dimension $d$, is denoted by
\Rrd.
\smallskip

A \emph{sand automaton} (SA) is a deterministic finite automaton
working on configurations. Each pile is updated synchronously,
according to a local rule which computes the variation of the pile
by means of the range. Formally, a SA is a
triple $\uple{d, r, f}$, where $d$ is the dimension, $r$ is the
\emph{radius} and $f: \Rrd \to \inter{-r}{r}$ is the \emph{local
rule} of the automaton. By means of the local rule, one can define
the \emph{global rule} $F: \C \to \C$ as follows
\[
\forall x\in\C,\,\forall i\in\Z^d,\quad F(x)_i=
\left\{\begin{array}{l@{\hspace{5mm}}l}
    x_i & \text{if $x_i=\pm\infty$}\enspace,\\
    x_i+f(R_r^i(x)) & \text{otherwise.}
  \end{array}\right.
\]
Remark that the radius $r$ of the automaton has three different
meanings: it represents at the same time the number of measuring
devices in every dimension of the range (number of piles in the
neighborhood), the precision of the measuring devices in the range,
and the highest return value of the local rule (variation of a pile).
It guarantees that there are only a finite number of ranges and return
values, so that the local rule has finite description.

The following example illustrates a sand automaton whose behavior
will be studied in Section \ref{sec:nil}. For more examples, we
refer to~\cite{CFM07}.
%---
\begin{example}[the automaton \SN]\label{ex:N}
This automaton destroys a configuration by collapsing all piles
towards the lowest one. It decreases a pile when there is a lower
pile in the neighborhood (see Figure~\ref{fig:N}). Let $\SN =
\para{1, 1, \fSN}$ of global rule $\FSN$ where
\[
  \forall a, b \in \intinf{-1}{1}, \quad \fSN(a, b) = \left\{
  \begin{array}{r@{\hspace{5mm}}l}
    -1 & \textrm{if $a < 0$ or $b < 0$}\enspace,\\
    0 & \textrm{otherwise.}
  \end{array}\right.
\]
\end{example}
\begin{figure}[!ht]
  \begin{center}
  \includegraphics[scale=0.7]{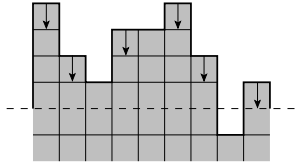}
  \caption{Illustration of the behavior of \SN.}
  \label{fig:N}
  \end{center}
\end{figure}
When no misunderstanding is possible, we identify a SA with its global
rule $F$. For any $k\in\Z^d$, we extend the definition of the
\emph{shift map} to \C, $\sigma^k:\C\to\C$ is defined by $\forall
x\in\C, \forall i\in\Z^d$, $\sigma^k(x)_i=x_{i+k}$. The \emph{raising
  map} $\rho: \C \to \C$ is defined by $\forall x\in\C, \forall i \in
\Z^d$, $\rho(x)_i = x_i + 1$. A function $F:\C\to\C$ is said to be
\emph{vertical-commuting} if $F\circ\rho=\rho\circ F$. A
function $F:\C\to\C$ is \emph{infinity-preserving} if for any
configuration $x\in\C$ and any vector $i\in\Z^d$, $F(x)_i=+\infty$ if
and only if $x_i=+\infty$ and $F(x)_i=-\infty$ if and only if
$x_i=-\infty$.

Remark that the raising map $\rho$ is the
  sand automaton of radius~$1$ whose local rule always returns~$1$. On
  the opposite, the horizontal shifts $\sigma_i$ are not sand
  automata: they destroy infinite piles by moving them, which is not
  permitted by the definition of the global rule.

\begin{theorem}[\cite{CF03,CFM07}]
  The class of SA is exactly the class of shift and
  vertical-commuting, infi\-ni\-ty-pre\-ser\-ving functions
  $F:\C\to\C$ which are continuous w.r.t. the metric $\dist^\prime$.
\end{theorem}
%---
\section{Topology and dynamics}\label{sec:topo}

In this section we introduce a compact topology on the SA
configuration space by means of a relation between SA and CA. With
this topology, a Hedlund-like theorem still holds and each SA turns
out to be homeomorphic to a CA acting on a specific subshift. We also
characterize CA whose action on this subshift represents a SA.
\NEW{Finally, we prove that equicontinuity is equivalent to ultimate
periodicity, and that expansivity is a very strong notion: there exist
no positively expansive SA.}

\subsection{A compact topology for SA configurations}

From \cite{CFM07}, we know that any SA of dimension $d$ can be
simulated by a suitable CA of dimension $d+1$ (and also any CA can
be simulated by a SA). In particular, a $d$-dimensional SA
configuration can be seen as a ($d+1$)-dimensional CA
configuration on the alphabet $A=\set{0,1}$. More precisely,
consider the function $\zeta: \C\to\set{0,1}^{\Z^{d+1}}$ defined
as follows
\[
\forall x\in\C, \quad \forall i\in\Z^d, \forall k\in\Z,\quad
\zeta(x)_{(i,k)}=
    \left\{\begin{array}{l@{\hspace{5mm}}l}
      1 & \text{if $x_i\ge k$}\enspace,\\
      0 & \text{otherwise.}
    \end{array}\right.
\]
A SA configuration $x\in\C$ is coded by the CA configuration
$\zeta(x)\in \set{0,1}^{\Z^{d+1}}$. Remark that $\zeta$ is an
injective function. \smallskip

\NEW{Consider the
  $(d+1)$-dimensional matrix $K\in\Mdd_(1,\ldots,1,2)$ such that
  $K_{1,\ldots,1,2}=1$ and $K_{1,\ldots,1,1}=0$.} With a little abuse
of notation, denote $S_K=S_{\{K\}}$ the subshift of configurations
that do not contain the pattern $K$.
\begin{proposition}\label{prop:sub}
The set $\zeta(\C)$ is the subshift $S_K$.
\end{proposition}
%---
\begin{proof}
Each $d$-dimensional SA configuration $x\in \C$ is coded by the
$(d+1)$-dimensional CA configuration $\zeta(x)$ such that for any
$i,h\in\Z^{d+1}, M^i_h(\zeta(x))\neq K$, then $\zeta(\C)\subseteq
S_K$.
Conversely, we can define a preimage by $\zeta$ for any $y\in S_K$, by $\forall i\in\Z^d,x_i=\sup\{k:y_{(i,k)}=1\}$. Hence $\zeta(\C)=S_K$.
\end{proof}
%---
Figure~\ref{fig:C-subshift} illustrates the mapping $\zeta$ and
the matrix $K=\left ( \begin{array}{c} 1 \\
0\end{array}\right)$ for the dimension $d=1$. The set of SA
configurations $\C = \Zt^{\Z}$ can be seen as the subshift
$S_K=\zeta(\C)$ of the CA configurations set $\set{0, 1}^{\Z^{2}}$.

%---
\begin{figure}[!ht]
  \begin{center}
  \subfigure[Valid configuration.]{
    \includegraphics{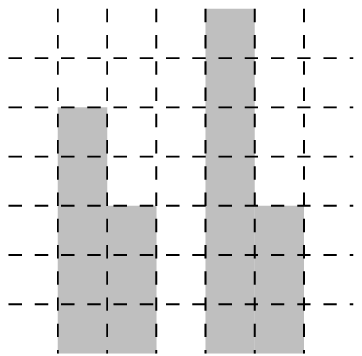}
    \label{fig:C-subshift-a}
  }\hspace{0cm}
  \subfigure[Invalid configuration.]{
    \includegraphics{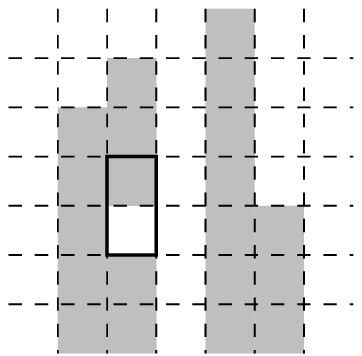}
    \label{fig:C-subshift-b}
  }
  \caption{The configuration from
Figure~\ref{fig:C-subshift-a} is valid, while the configuration
from Figure~\ref{fig:C-subshift-b} contains the forbidden matrix
$K$: there is a ``hole''.} \label{fig:C-subshift}
  \end{center}
\end{figure}
%---

%---
\begin{definition}\label{newdist}
The distance $\dist:\C\times\C\to\re_+$ is defined as follows:
\[
\forall x,y\in\C, \quad
\dist(x,y)=\dist_T(\zeta(x),\zeta(y))\enspace.
\]
\end{definition}
%---
In other words, the (well defined) distance $\dist$ between two
configurations $x,y\in C$ is nothing but the Tychonoff distance
between the configurations $\zeta(x), \zeta(y)$ in the subshift
$S_K$.
The corresponding metric topology is the $\{0,1\}^{\Z^{d+1}}$ product
topology induced on $S_K$.

\begin{remark}\label{rem-proj}
  Note that this topology does not coincide with the topology obtained
  as countable product of the discrete topology on $\Zt$.  Indeed, for
  any $i\in\Z^d$, the $i^{\text{th}}$ projection $\pi_i:\C\to\Zt$
  defined by $\pi_i(x)=x_i$ is not continuous in any configuration $x$
  with $x_i=\pm\infty$. However, it is continuous in all configurations
  $x$ such that $x_i\in\Z$, since $\forall k\in\Z,\forall x,y \in \C$,
  conditions $\pi_i(x)=k$ and $\dist(x,y)\le2^{-\max(|i|,k)}$ imply that
  $\pi_i(y)=k$.
%\TODO{topo}{produit de celle ou les infinis sont pas isoles}
\end{remark}

By definition of this topology, if one considers $\zeta$ as a map from
$\C$ onto $S_K$, $\zeta$ turns out to be an isometric
homeomorphism between the metric spaces $\C$ (endowed with $\dist$) and
$S_K$ (endowed with $\dist_T$). As an immediate consequence, the following
results hold.
\begin{proposition}\label{prop:C-comp-etc}
The set \C
is a compact and totally disconnected space where the open
balls are clopen (\ie closed and open) sets.
\end{proposition}
%---
\begin{proposition}\label{prop:C-perfect}
  The space \C is perfect.
\end{proposition}
\begin{proof}
   Choose an arbitrary configuration $x\in \C$.
   For any $n\in\N$, let $l \in \Z^d$ such that $|l| = n$.
   We build a configuration $y\in \C$,
   equal to $x$ except at site $l$, defined as follows
   \[
   \forall j\in\Z^d\setminus\set{l},\;y_j=x_j
   \quad \text{and} \quad
   y_l = \left\{
     \begin{array}{c@{\hspace{5mm}}l}
       1 & \text{if $x_l = 0$}\enspace,\\
       0 & \text{otherwise.}
     \end{array}\right.
   \]
  By Definition \ref{newdist}, $\dist(y,x) = 2^{-n}$.
\end{proof}
%---
Consider now the following notion.
\begin{definition}[ground cylinder]
For any configuration $x\in \C$, for any $r\in\n$, and for any
$i\in\z^d$, the \emph{ground cylinder} of $x$ centered on $i$ and
of radius $r$ is the $d$-dimensional matrix $C^i_r(x) \in \Mrd$ defined by
\[
  \forall k \in \inter{1}{2r+1}^d, \quad \left(C^i_r(x)\right)_k =
  \obrx{i+k-r-1} \enspace.
\]
\end{definition}
For example in dimension $1$,
\[
  C^i_r(x) = \parto{\obrx{i-r}, \ldots, \obrx{i}, \ldots, \obrx{i+r}} \enspace.
\]

Figure~\ref{fig:cylinders} illustrates top cylinders and ground
cylinders in dimension~$1$. Remark that the content of the two
kinds of cylinders is totally different.
\begin{figure}[!ht]
  \begin{center}
  \subfigure[Top cylinder centered on $x_i=4$:\newline $\oldc^i_r(x) = (+1, -\infty, -3, \mathbf{4}, -2, -2, +1)$.]{
    \includegraphics[scale=0.7]{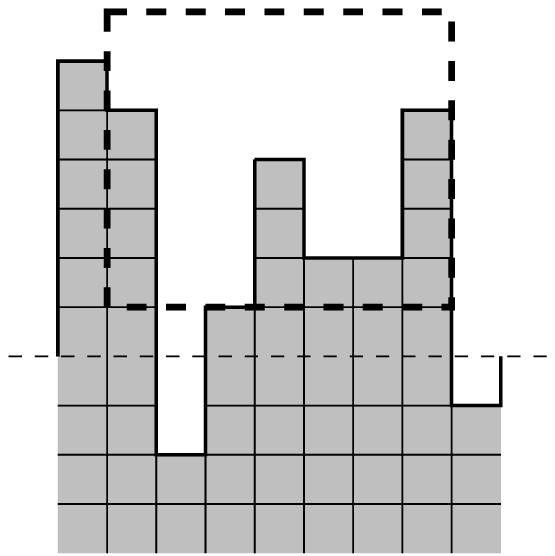}
    \label{fig:top-cylinder}
  }\hspace{5mm}
  \subfigure[Ground cylinder, at height~$0$:\newline $C^i_r(x) = (+\infty, -2, +1, \mathbf{+\infty}, +2, +2, +\infty)$.]{
    \includegraphics[scale=0.7]{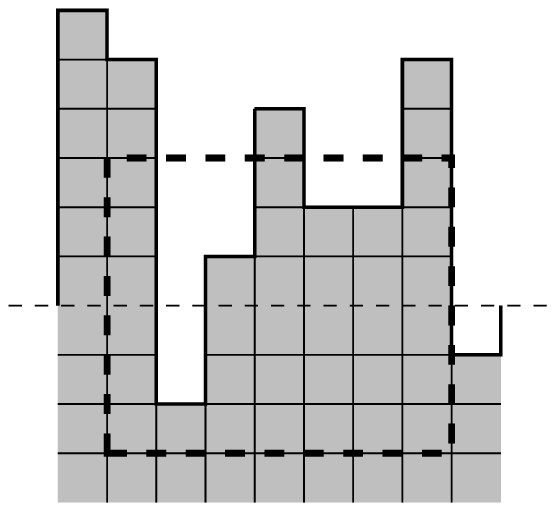}
    \label{fig:ground-cylinder}
  }
  \caption{Illustration of the two notions of cylinders on the same configuration, with radius~$3$, in dimension~$1$.}
  \label{fig:cylinders}
  \end{center}
\end{figure}

From Definition \ref{newdist}, we obtain the following expression of distance $\dist$ by means of ground cylinders.
\begin{remark}
For any pair of configurations $x,y\in\C$, we have
\[
\dist(x,y)=2^{-k} \quad \text{where}\quad k=\min\set{r\in\n:
C^0_r(x)\neq C^0_r(y)} \enspace.
\]
\end{remark}
%---
As a consequence, two configurations $x,y$ are compared by putting
boxes (the ground cylinders) at height~$0$ around the
corresponding piles indexed by $0$. The integer $k$ is the size of
the smallest cylinders in which a difference appears between $x$
and $y$. This way of calculating the distance $\dist$ is similar to
the one used for the distance $\dist^\prime$, with the difference
that the measuring devices and the cylinders are now located at
height $0$. This is slightly less intuitive than the distance
$\dist^\prime$, since it does not correspond to the definition of
the local rule. However, this fact is not an issue all the more
since the configuration space is compact and the representation
theorem still holds with the new topology
(Theorem~\ref{th:charac}).

\subsection{SA as CA on a subshift}
Let $(X,m_1)$ and $(Y,m_2)$ be two metric spaces. Two functions
$H_1:X\to X$, $H_2:Y\to Y$ are (topologically) \emph{conjugated} if
there exists a homeomorphism $\eta:X\to Y$ such that
$H_2\circ\eta=\eta\circ H_1$.

We are going to show that any SA is conjugated to some restriction of
a CA. Let $F$ a $d$-dimensional SA of radius $r$ and local rule $f$.
Let us define the ($d+1$)-dimensional CA $G$ on the alphabet
$\{0,1\}$, with radius $2r$ and local rule $g$ defined as follows (see
\cite{CFM07} for more details). \NEWa{Let $M \in
  \mathcal{M}^{d+1}_{\mathbf{4r+1}}$ be a matrix on the finite alphabet
  $\set{0,1}$ which does not contain the pattern $K$. If there is a
  $j\in[r+1,3r]$ such that $M_{(2r+1,\ldots,2r+1,j)}=1$ and
  $M_{(2r+1,\ldots,2r+1,j+1)}=0$, then let $R \in \Rrd$ be the range taken from
  $M$ of radius $r$ centered on $(2r+1,\ldots,2r+1,j)$. See
  figure~\ref{fig:SA-CA} for an illustration of this construction in
  dimension~$d=1$.}
\begin{figure}[!ht]
  \begin{center}
    \includegraphics[scale=0.7]{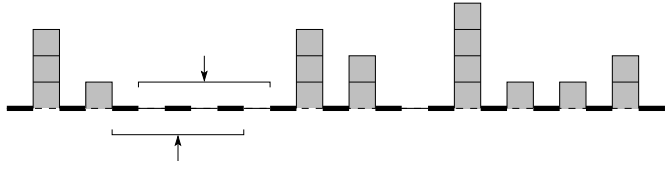}
    \caption{\NEWa{Construction of the local rule $g$ of the CA from the
      local rule $f$ of the SA, in dimension~$1$. A range $R$ of radius $r$ is
      associated to the matrix $M$ of order $\mathbf{4r+1}$.}}
    \label{fig:SA-CA}
  \end{center}
\end{figure}

\NEWa{The new central value depends on the height $j$ of the central column
plus its variation. Therefore, define $g(M)=1$ if $j+f(R)\geq 0$,
$g(M) = 0$ if $j+f(R)<0$, or $g(M) = M_{(2r+1, \ldots, 2r+1)}$
(central value unchanged) if there is no such $j$.}

\smallskip

The following diagram commutes:
\begin{equation}\label{cd}
\begin{CD}
  \C  @>F>> & \C\\
   @V{\zeta}VV&@VV{\zeta}V\\
  S_K @>>G> & S_K
\end{CD}\enspace,
\end{equation}
\ie $G\circ\zeta=\zeta\circ F$. As an immediate consequence, we
have the following result.
\begin{proposition}\label{prop:conju}
Any $d$-dimensional SA $F$ is topologically conjugated to a
suitable $(d+1)$-dimensional CA $G$ acting on $S_K$.
\end{proposition}

\NEWa{Being a dynamical submodel, SA share properties with
CA, some of which are proved below. However, many results which
are true for CA are no longer true for SA; for instance,
injectivity and bijectivity are not equivalent, as proved in
\cite{cervelle05}. Thus, SA deserve to be considered as a new
model.}

\begin{corollary}\label{cor:unifcont}
The global rule $F:\C\to\C$ of a SA is uniformly continuous w.r.t distance $\dist$.
\end{corollary}
\begin{proof}
Let $G$ be the global rule of the CA which simulates the given SA.
Since the diagram~\eqref{cd} commutes and $\zeta$ is a
homeomorphism, $F=\zeta^{-1}\circ G\circ \zeta$. Since $G$ is a
continuous map and, by Proposition~\ref{prop:C-comp-etc}, $\C$ is
compact, then the thesis is obtained.
\end{proof}
For every $a\in\Z$, let $P_a=\pi_0^{-1}(\{a\})$ be the clopen (and compact) set of all configurations $x\in\C$ such that $x_{\nv}=a$.
\begin{lemma}\label{lem:l}
  Let $F:\C\to\C$ be a continuous and infinity-preserving
  map. There exists an integer $l\in\n$ such that for any
  configuration $x\in P_0$ we have $|F(x)_0|\leq l$.
\end{lemma}
\begin{proof}
Since $F$ is continuous and infinity-preserving, the set $F(P_0)$
is compact and included in $\pi_0^{-1}(\Z)$.
  From Remark~\ref{rem-proj}, $\pi_0$
is continuous on the set $\pi_0^{-1}(\Z)$ and in particular it is
continuous on the compact $F(P_0)$. Hence $\pi_0(F(P_0))$ is a
compact subset of $\Zt$ containing no infinity, and therefore it is
included in some interval $[-l,l]$, where $l\in\n$.
\end{proof}
%---

\begin{theorem}\label{th:charac}
  A mapping $F:\C\to\C$ is the global transition rule of a sand
  automaton if and only if all the
  following statements hold
\begin{enumerate}[$(i)$]
 \item $F$ is (uniformly) continuous w.r.t the distance $\dist$;
 \item\label{item:comsigma} $F$ is shift-commuting;
 \item\label{item:comrho} $F$ is vertical-commuting;
 \item\label{item:inf} $F$ is infinity-preserving.
\end{enumerate}
\end{theorem}

\begin{proof}
Let $F$ be the global rule of a SA. By definition of SA, $F$ is
shift-commuting, vertical-commuting and infinity-preserving.
From Corollary~\ref{cor:unifcont}, $F$ is also uniformly
continuous.

\smallskip

Conversely, let $F$ be a continuous map which is shift-commuting,
vertical-commuting, and infinity-preserving. By compactness of
the space
$\C$, $F$ is also uniformly continuous. % pourquoi le commenter ?
Let $l\in\n$ be the integer given by Lemma~\ref{lem:l}. Since $F$ is
uniformly continuous, there exists an integer $r\in\n$ such that
\[
\forall x,y\in\C \quad C_r^0(x)=C_r^0(y) \Rightarrow
C_l^0(F(x))=C_l^0(F(y)) \enspace.
\]
We now construct the local rule $f:\Rrd\to \inter{-r}{r}$ of the
automaton. For any input range $R\in\Rrd$, set $f(R)=F(x)_0$, where
$x$ is an arbitrary configuration of $P_0$ such that \NEW{$\forall k
  \in \inter{1}{2r+1}$, $k \ne r+1$, $\obrx{k-r-1}=R_k$}. Note that
the value of $f(R)$ does not depend on the particular choice of the
configuration $x\in P_0$ such that $\forall k\neq r+1$, $\obrx{k-r-1}=R_k$.
Indeed, Lemma~\ref{lem:l} and uniform continuity together ensure that
for any other configuration $y\in P_0$ such that $\forall k\neq r+1$,
$\obry{k-r-1}=R_k$, we have $F(y)_0=F(x)_0$, since $\oblFx{0}=\oblFy{0}$
and $|F(y)_0|\leq l$. Thus the rule $f$ is well defined.

We now show that $F$ is the global mapping of the sand automaton of
radius $r$ and local rule $f$. Thanks to $(\ref{item:inf})$, it is sufficient
to prove that for any $x\in\C$ and for any $i\in\z^d$ with $|x_i|\neq \infty$,
we have $F(x)_i=x_i+f\parto{R^i_r(x)}$. By $(\ref{item:comsigma})$ and $(\ref{item:comrho})$, for any $i \in \Z^d$ such that $|x_i|\neq \infty$, it holds that
\begin{align*}
    F(x)_i &= \left[\rho^{x_i} \circ \sigma^{-i}\left(F(\sigma^i \circ \rho^{-x_i}(x))\right)\right]_i \\
               &= x_i + \left[\sigma^{-i}\left(F(\sigma^i \circ \rho^{-x_i}(x))\right)\right]_i \\
               &= x_i + \left[F(\sigma^i \circ \rho^{-x_i}(x))\right]_0 \enspace.
\end{align*}
Since $\sigma^i \circ \rho^{-x_i}(x) \in P_0$, we have by definition of $f$
\[
  F(x)_i = x_i + f\left(R_r^0(\sigma^i \circ \rho^{-x_i}(x))\right) \enspace.
\]
Moreover, by definition of the range, for all $k \in \inter{1}{2r+1}^d$,
\[
  R_r^0(\sigma^i \circ \rho^{-x_i}(x))_k =
  \beta_r^{[\sigma^i \circ \rho^{-x_i}(x)]_0}(\sigma^i \circ \rho^{-x_i}(x)_k) = \beta_r^0(x_{i+k}-x_i) = \beta_r^{x_i}(x_{i+k}) \enspace,
\]
hence $R_r^0(\sigma^i \circ \rho^{-x_i}(x)) = R_r^i(x)$, which leads to $F(x)_i = x_i + f\parto{R^i_r(x)}$.
\end{proof}

%----
We now deal with the following question: given a
($d+1$)-dimensional CA, does it represent a $d$-dimensional SA, in
the sense of the conjugacy expressed by diagram~\ref{cd}? In order
to answer to this question we start to express the condition under
which the action of a CA $G$ can be restricted to a subshift
$S_\cF$, \ie, $G(S_\cF)\subseteq S_\cF$ (if this fact holds, the
subshift $S_\cF$ is said to be $G$-invariant).

\NEW{\begin{lemma}\label{lemm:subshift-inv}
  Let $G$ and $S_\cF$ be a CA and a subshift of finite type,
  respectively. The condition $G(S_\cF)\subseteq S_\cF$ is satisfied
  iff for any $U\in\ccL(S_\cF)$ and any $H\in\cF$ of the same order than $g(U)$, it holds that $g(U)\neq H$.
\end{lemma}
%----
\begin{proof}
Suppose that $G(S_\cF)\subseteq S_\cF$. Choose arbitrarily
$H\in\cF$ and $U\in\ccL(S_\cF)$, with $g(U)$ and $H$ of the same
order. Let $x\in S_\cF$ containing the matrix $U$. Since $G(x)\in
S_\cF$, then $g(U)\in\ccL(S_\cF)$, and so $g(U)\neq H$. Conversely,
if $x\in S_\cF$ and $G(x)\notin S_\cF$, then there exist
$U\in\ccL(S_\cF)$ and $H\in\cF$ with $g(U)= H$.
\end{proof}}
%---
\NEWa{The following proposition gives a sufficient and necessary
  condition under which the action of a CA $G$ on
  configurations of the $G$-invariant subshift $S_K = \C$ preserves
  any column whose cells have the same value.}

\NEW{\begin{lemma}\label{lemm:infinite-pres} Let $G$ be a
    $(d+1)$-dimensional CA with state set $\{0,1\}$ and $S_K$ be the
    subshift representing SA configurations. The following two statements are
    equivalent:
\begin{enumerate}[$(i)$]
\item for any $x\in S_K$ with $x_{(0,\ldots, 0,i)}=1$ (resp.,
$x_{(0,\ldots, 0, i)}=0$) for all $i\in\Z$, it holds that
 $G(x)_{(0,\ldots, 0,i)}=1$ (resp., $G(x)_{(0,\ldots, 0,i)}=0$) for all
 $i\in\Z$.
\item for any matrix $U\in\Mrd\cap \ccL(S_K)$ with  $U_{(r+1,\ldots, r+1,
k)}=1$ (resp., $U_{(r+1,\ldots, r+1,k)}=0$) and any $k\in \inter{1}{2r+1}$, it
holds that $g(U)=1$ (resp., $g(U)=0$).
 \end{enumerate}
\end{lemma}
\begin{proof}
  Suppose that $(1)$ is true. Let $U\in\Mrd\cap \ccL(S_K)$ be a matrix
  with $U_{(r+1,\ldots, r+1, k)}=1$ and let $x\in S_K$ be a
  configuration such that $x_{(0,\ldots, 0,i)}=1$ for all $i\in\Z$ and
  $\Mr^{-\rr}(x)=U$. Since $G(x)_{(0,\ldots, 0,i)}=1$ for all
  $i\in\Z$, and $\Mr^0(x)=U$, then $g(U)=1$. Conversely, let $x\in
  S_K$ with $x_{(0,\ldots, 0,i)}=1$ for all $i\in\Z$. By
  shift-invariance, we obtain $G(x)_{(0,\ldots, 0,i)}=1$ for all
  $i\in\Z$.
\end{proof}
}

Lemmas~\ref{lemm:subshift-inv} and~\ref{lemm:infinite-pres}
immediately lead to the following conclusion.
\begin{proposition}
It is decidable to check whether a given ($d+1$)-dimensional CA
corresponds to a $d$-dimensional SA.
\end{proposition}

\subsection{Some dynamical behaviors}
\NEWa{SA are very interesting dynamical systems, which in some sense
  ``lie'' between $d$-dimensional and $d+1$-dimensional CA. Indeed, we
  have seen in the previous section that the latter can simulate
  $d$-dimensional SA, which can, in turn, simulate $d$-dimensional CA.
  For the dimension $d=1$, a classification of CA in terms of their
  dynamical behavior was given in \cite{ku97}. Things are very
  different as soon as we get into dimension $d=2$, as noted in
  \cite{noexpansive2DCA,theyssier-sablik-arxiv}. The question is now
  whether the complexity of the SA model is closer to that of the
  lower or the higher-dimensional CA.}

\ignore{
One-dimensional SA are very interesting models, which complexity
 lies somewhere between one-dimensional and two-dimensional CA.
Indeed we have seen in the previous section that the latter can
simulate SA, and it was shown in \cite{CFM07} that SA could
simulate the former. A classification of one-dimensional cellular
automata in terms of their dynamical behavior was given in
\cite{ku97}. Things are very different as soon as we get into the
second dimension, as noted in
\cite{noexpansive2DCA,theyssier-sablik-arxiv}. The question is now
whether the complexity of the SA model is closer to that of the
lower or the higher-dimensional CA.
}
Let $(X,m)$ be a metric space and let $H:X\to X$ be a continuous
application.
An element
$x\in X$ is an \emph{equicontinuity} point for $H$ if for any
$\varepsilon>0$, there exists $\delta>0$ such that for all $y\in
X$, $m(x,y)<\delta$ implies that $\forall n\in\n$,
$m(H^n(x),H^n(y))<\varepsilon$. The map $H$ is
\emph{equicontinuous} if for any $\varepsilon>0$, there exists
$\delta>0$ such that for all $x,y\in X$, $m(x,y)<\delta$ implies
that $\forall n\in\n$, $m(H^n(x), H^n(y))<\varepsilon$.
If $X$ is
compact, $H$ is equicontinuous iff all elements of $X$ are
equicontinuity points.
An element $x\in X$ is \emph{ultimately
periodic} for $H$ if there exist two integers $n\geq 0$ (the
preperiod) and $p>0$ (the period) such that $H^{n+p}(x)=H^n(x)$.
$H$ is \emph{ultimately periodic} if there exist $n\ge0$ and $p>0$
such that $H^{n+p}=H^n$.
$H$ is \emph{sensitive} (to the initial conditions) if there is a
constant $\varepsilon>0$ such that for all points $x\in X$ and all
$\delta>0$, there is a point $y\in X$ and an integer $n\in\n$ such
that $m(x,y)<\delta$ but $m(F^n(x),F^n(y))>\varepsilon$. $H$ is
\emph{positively expansive} if there is a constant $\varepsilon>0$ such
that for all distinct points $x,y\in X$, there exists $n\in\n$
such that $m(H^n(x),H^n(y))>\varepsilon$.

The topological conjugacy between a SA and some CA acting on the
special subshift $S_K$ helps to adapt some properties of CA. In
particular, the following characterization of equicontinuous CA
can be adapted from Theorem~4 of \cite{ku97}.

\begin{proposition}\label{prop:ult-equiv}
If $F$ is a SA, then the following statements are equivalent:
\begin{enumerate}
\item\label{i-feq} $F$ is equicontinuous.
\item\label{i-fup} $F$ is ultimately periodic.
\item\label{i-pup} \NEWa{All configurations of \C are ultimately periodic for $F$}.
\end{enumerate}
\end{proposition}
\begin{proof}
\ref{i-pup}$\Rightarrow$\ref{i-fup}: For any $n\geq 0$ and $p>0$,
let $D_{n,p}=\set{x: F^{n+p}(x) = F^n(x)}$. Remark that
$\C=\bigcup_{n,p\in\N}D_{n,p}$ is the union of these closed
subsets. As $\C$ is complete of nonempty interior, by the Baire
Theorem, there are integers $n, p\in\N$ for which the set
$D_{n,p}$ has nonempty interior. Hence the conjugate image
$\zeta(D_{n,p})$ has nonempty interior too, and it can easily be
seen that it is a subshift. It is known that the only subshift
with nonempty interior is the full space; hence $D_{n,p}=\C$.
\\
\ref{i-fup}$\Rightarrow$\ref{i-pup}: obvious.
\\
\ref{i-fup}$\Rightarrow$\ref{i-feq}: Let $F$ be ultimately
periodic with $F^{n+p}=F^n$ for some $n\geq 0$, $p>0$. Since $F,
F^2,\ldots, F^{n+p-1}$ are uniformly continuous maps, for any
$\varepsilon>0$ there exists $\delta>0$ such that for all
$x,y\in\C$ with $\dist(x,y)<\delta$, it holds that $\forall
q\in\n$, $q<n+p$, $\dist(F^q(x),F^q(y))<\varepsilon$. Since for
any $t\in\n$ $F^t$ is equal to some $F^q$ with $q<n+p$, the map
$F$ is equicontinuous.
\\%\item[
\ref{i-feq}$\Rightarrow$\ref{i-fup}: For the sake of simplicity,
we give the proof for a given one-dimensional equicontinuous SA
$F$. Let $G$ be the global rule of the two-dimensional CA whose
action on $S_K$ is conjugated to $F$. By Definition~\ref{newdist},
and since the diagram~\ref{cd} commutes, the map $G:S_K\to S_K$ is
equicontinuous w.r.t. $\dist_T$. So, for $\varepsilon=1$, there
exists $l\in\n$ such that for all $x,y\in S_K$, if
$\Ml^{-\lu}(x)=\Ml^{-\lu}(y)$, then for all $t\in\n$,
$G^t(x)_{0}=G^t(y)_{0}$. Consider now configurations $\zeta(c)$,
where $c\in\{-\infty, +\infty\}^{\z}$ has either the form
$(\ldots, -\infty,-\infty, +\infty, +\infty,\ldots)$ or $(\ldots,
+\infty,+\infty, -\infty, -\infty,\ldots)$. Since every $\zeta(c)$ are
ultimately periodic (with pre\-pe\-riod $n=0$ and period $p=1$) and
$G$ is equicontinuous, for any $k\in\z^2$ and any $y\in S_K$ with
$\Ml^{k-\lu}(y)=\Ml^{k-\lu}(\zeta(c))$, it holds that the sequence
$\{G^t(y)_{k}\}_{t\in\n}$ is ultimately periodic. For any
$U\in\ccL(S_K)\cap \ensuremath{\mathcal{M}^2_{\mathbf{2l+1}}}$,
let $x^U$ be the configuration such that $\Ml^{-\lu}(x)=U$,
$x_{(i,j)}=0$ if $-l\leq i\leq l$ and $j>l$, and $x_{(i,j)}=1$
otherwise. Except for the finite central region, $x^U$ is made by the
repetition of a finite number of matrices appearing inside
configurations $\zeta(c)$. Hence, $x^U$ is an ultimately periodic
configuration with some preperiod $n_U$ and period $p_U$. Then,
for any $y\in S_K$ with $\Ml^{-\lu}(y)=U$, the sequence
$\{G^t(y)_{0}\}_{t\in\n}$ is ultimately periodic with preperiod
$n_U$ and period $p_U$. Set $n=\max\{n_U: U\in\ccL(S_K)\cap
\ensuremath{\mathcal{M}^2_{\mathbf{2l+1}}}\}$ and
$p=\text{lcm}\{p_U: U\in\ccL(S_K)\cap
\ensuremath{\mathcal{M}^2_{\mathbf{2l+1}}}\}$ where lcm is the
least common multiple. Thus, for any configuration $z\in S_K$, we
have that $G^n(z)_0=G^{n+p}(z)_0$. By shift-invariance, we obtain
$\forall k\in\z^2$, $G^n(z)_k=G^{n+p}(z)_k$. Concluding, $G$ is
ultimately periodic and then $F$ is too.
\end{proof}

In \cite{ku97} is presented a classification of CA into four
classes: equicontinuous CA, non equicontinuous CA admitting an
equicontinuity configuration, sensitive but not positively
expansive CA, positively expansive CA. This classification is no
more relevant in the context of SA since the class of positively
expansive SA is empty. This result can be related to the absence
of positively expansive two-dimensional CA (see
\cite{noexpansive2DCA}), though the proof is much different.
\begin{proposition}
There are no positively expansive SA.
\end{proposition}
\begin{proof}
Let $F$ a SA and $\delta=2^{-k}>0$. Take two distinct configurations $x,y\in\C$ such that $\forall i\in[-k,k],x_i=y_i=+\infty$. By infinity-preservingness, we get $\forall n\in\n,\forall i\in[-k,k],F^n(x)_i=F^n(y)_i=+\infty$, hence $d(F^n(x),F^n(y))<\delta$.
\end{proof}

An important open question in the dynamical behavior of SA is the
existence of non-sensitive SA without any equicontinuity
configuration. An example for two-dimensional CA is given in \cite{theyssier-sablik-arxiv}, but their method can hardly be adapted for SA.
This could lead to a classification of SA into four classes: equicontinuous, admitting an equicontinuity configuration (but not equicontinuous), non-sensitive without equicontinuity configurations, sensitive.

Another issue is the decidability of these classes. In
\cite{cervelle05}, the undecidability of SA ultimate periodicity was
proved on the particular subsets of finite and periodic configurations. It follows directly that equicontinuity on these subsets is undecidable. The question is still open for the whole configuration space \C.

\section{The nilpotency problem}\label{sec:nil}

In this section we give a definition of nilpotency for SA. Then, we
prove that nilpotency behavior is undecidable
(Theorem~\ref{th:nil-undec}).

\subsection{Nilpotency of CA}

Here we recall the basic definitions and properties of nilpotent CA.
Nilpotency is among the simplest dynamical behavior that an automaton
may exhibit. Intuitively, an automaton defined by a local rule and
working on configurations (either $\C$ or $\az$) is nilpotent if it
destroys every piece of information in any initial configuration,
reaching a common constant configuration after a while. For CA, this
is formalized as follows.

\begin{definition}[CA nilpotency~\cite{CPY89,K92}]\label{def:canil}
A CA $G$ is nilpotent if
\begin{equation*}
\exists c\in A, \quad \exists N\in\N \quad \forall x\in\az, \quad
\forall n\geq N, \quad G^n(x)=\cst{c}\enspace.
\end{equation*}
\end{definition}

Remark that in a similar way to the proof of
Proposition~\ref{prop:ult-equiv}, \NEW{Definition~\ref{def:canil} can
be restated as follows: a CA is nilpotent if and only if
it is nilpotent for all initial configurations.}
\ignore{
\begin{proposition}\label{prop:N-spreading}
  A CA $G$
  is nilpotent if and only if
  \[
    \exists c\in A, \quad \forall x \in \az, \quad \exists N \in \N,
    \quad \forall n\geq N, \quad G^n(x) = \cst{c}\enspace.
  \]
\end{proposition}
}

%\TODO{arxiv: all 1D CA?}
Spreading CA have the following stronger characterization.

\begin{proposition}[\cite{CG07}]
  A CA $G$, with spreading state~$0$,
  is nilpotent iff for every $x \in \az$,
  there exists $n \in \N$ and $i \in \Z^d$ such that
  $G^n(x)_i = 0$ (\ie $0$ appears in the evolution of every
  configuration).
\end{proposition}

The previous result immediately leads to the following
equivalence.

\begin{corollary}\label{cor:spreading-nil-equiv}
  A CA of global rule $G$, with spreading state~$0$, is nilpotent if and only if for all configurations $x \in \az$, $\lim_{n\to\infty} \dist_T(G^n(x), \cst{0}) = 0$.
\end{corollary}

Recall that the CA nilpotency is undecidable~\cite{K92}. Remark
that the proof of this result also works for the restricted class
of spreading CA.

\begin{theorem}[\cite{K92}]\label{thm:spreading-nil-indec}
  For a given state~$s$, it is undecidable to know whether a
  cellular automaton with spreading state~$s$ is nilpotent.
\end{theorem}

\subsection{Nilpotency of SA}

A direct adaptation of Definition~\ref{def:canil} to SA is vain.
Indeed, assume $F$ is a SA of radius $r$. For any
$k\in\Z^d$, consider the configuration $x^k\in\B$ defined by
$x^k_0=k$ and $x^k_i=0$ for any $i\in\Z^d\backslash\{0\}$.
Since the pile of height $k$ may decrease at most by $r$ during one
step of evolution of the SA, and the other piles may increase at
most by $r$, $x^k$ requires at least $\lceil k/2r \rceil$ steps to
reach a constant configuration. Thus, there exists no common integer
$n$ such that all configurations $x^k$ reach a constant configuration
in time $n$.
This is a major
difference with CA, which is essentially due to the unbounded set of
states and to the infinity-preserving property.

Thus, we propose to label as nilpotent the SA which make every pile
approach a constant value, but not
necessarily reaching it ultimately.
\NEW{This nilpotency notion, inspired by
  Proposition~\ref{cor:spreading-nil-equiv}, is formalized as follows
  for a SA $F$:}
\begin{equation*}
  \exists c\in \Z, \quad \forall x\in\C, \quad \lim_{n\to\infty} \dist(F^n(x), \cst{c}) = 0 \enspace.
\end{equation*}
Remark that $c$ shall not be taken in the full state set $\Zt$, because
allowing infinite values for $c$ would not correspond to the intuitive
idea that a nilpotent SA ``destroys'' a configuration (otherwise, the
raising map would be nilpotent). Anyway, this definition is not
satisfying because of the vertical commutativity: two configurations
which differ by a vertical shift reach two different configurations,
and then no nilpotent SA may exist. A possible way to work around this
issue is to make the limit configuration depend on the initial one:
\begin{equation*}
  \forall x\in\C, \quad \exists c\in \Z, \quad \quad \lim_{n\to\infty} \dist(F^n(x), \cst{c}) = 0 \enspace.
\end{equation*}

\ignore{Note that we can assume that $c$ does
  not depend on $n$, by restricting to SA such that $f(O_d)=0$, where
  $O_d$ is the $d$-dimensional matrix which has only null
  coefficients. }Again, since SA are infinity-preserving, an infinite
pile cannot be destroyed (nor, for the same reason, can an infinite
pile be built from a finite one).  Therefore nilpotency has to involve
the configurations of $\Z^{\Z^d}$, \ie the ones without infinite
piles. Moreover, every configuration $x\in \Z^{\Z^d}$ made of regular
steps (\ie in dimension~$1$, for all $i \in \Z$, $x_i-x_{i-1} =
x_{i+1}-x_i$) is invariant by the SA rule (possibly composing it with
the vertical shift). So it cannot reach nor approach a constant
configuration. Thus, the larger reasonable set on which nilpotency
might be defined is the set of bounded configurations \B. This leads
to the following formal definition of nilpotency for SA.

\begin{definition}[SA nilpotency]\label{def:sanil}
A SA $F$ is nilpotent if and only if
\begin{equation*}
  \forall x\in\B, \quad \exists c\in \Z, \quad \quad \lim_{n\to\infty} \dist(F^n(x), \cst{c}) = 0 \enspace.
\end{equation*}
\end{definition}

The following proposition shows that the class of nilpotent SA is
nonempty.

\begin{proposition}
  The SA \SN from Example~\ref{ex:N} is nilpotent.
\end{proposition}
\begin{proof}
  Let $x \in \B$, let $i \in \Z$ such that for
  all $j \in \Z$, $x_j \geq x_i$. Clearly, after
  $x_{i+1} - x_i$ steps, $\FSN^{x_{i+1} - x_i}(x)_{i+1}
  = \FSN^{x_{i+1} - x_i}(x)_{i} = x_i$. By immediate
induction, we obtain that for all $j \in \Z$ there exists $n_j \in
\N$ such that $\FSN^{n_j}(x)_{j} = x_i$, hence
$\lim_{n\to\infty}\dist(\FSN^n(x), \cst{x_i}) = 0$.
\end{proof}

Similar nilpotent SA can be constructed with any radius and in any
dimension.

\subsection{Undecidability}

The main result of this section is that SA nilpotency is
undecidable (Theorem~\ref{th:nil-undec}), by reducing the
nilpotency of spreading CA to it. This emphasizes the fact that the
dynamical behavior of SA is very difficult to predict. We think
that this result might be used as the reference undecidable
problem for further questions on SA.

\medskip
\noindent \textbf{Problem} \textbf{Nil}\\
\noindent \qquad\textsc{instance}: a SA $\A=\langle d, r, \lambda \rangle$;\\
\noindent \qquad\textsc{question}: is \A nilpotent?

\begin{theorem}\label{th:nil-undec}
  The problem \textbf{Nil} is undecidable.
\end{theorem}
\begin{proof}
  This is proved by reducing \textbf{Nil} to the nilpotency of spreading cellular automata. Remark that it is sufficient to show the result in dimension~$1$. Let \S be a spreading cellular automaton $\S=\uple{A, 1, s, g}$ of global rule $G$, with finite set of integer states $A \subset \N$ containing the spreading state~$0$. We simulate \S with the sand automaton $\A=\uple{1, r = \max (2s, \max A), f}$ of global rule $F$ using the following technique, also developed in~\cite{CFM07}. Let $\xi: A^\Z \to \B$ be a function which inserts markers every two cells in the CA configuration to obtain a bounded SA configuration. These markers allow the local rule of the SA to know the absolute state of each pile and behave as the local rule of the CA. To simplify the proof, the markers are put at height~$0$ (see Figure~\ref{fig:CA-SA}):
\[
  \forall y \in A^\Z, \forall i \in \Z, \quad \xi(y)_i = \left\{\begin{array}{l@{\hspace{5mm}}l}
    0\:\text{(marker)} & \text{if $i$ is odd}\enspace,\\
    y_{i/2} & \text{otherwise.}
    \end{array}\right.
\]
This can lead to an ambiguity when all the states in the
neighborhood of size $4s+1$ are at state $0$, as shown in the
picture. But as in this special case the state~$0$ is quiescent
for $g$, this is not a problem: the state~$0$ is preserved, and
markers are preserved.
\begin{figure}[!ht]
  \begin{center}
  \includegraphics[scale=0.7]{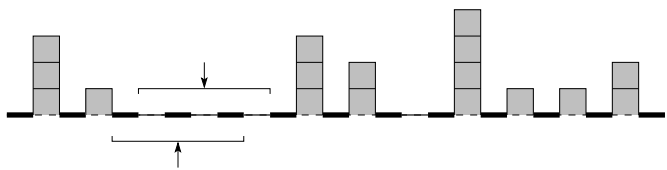}
  \caption{Illustration of the function $\xi$ used in the simulation
    of the spreading CA \S by \A. The thick segments are the markers
    used to distinguish the states of the CA, put at height $0$. There
    is an ambiguity for the two piles indicated by the arrows: with a
    radius 2, the neighborhoods are the same, although one of the
    piles is a marker and the other the state~$0$.}
  \label{fig:CA-SA}
  \end{center}
\end{figure}

The local rule $f$ is defined as follows, for all ranges $R \in
\mathcal{R}^1_r$,
\begin{equation}\label{eq:fvalid}
f(R) = \left\{\begin{array}{lr@{}}
    0 & \text{if $R_{-2s+1}, R_{-2s+3}, \ldots, R_{-1}, R_{1}, \ldots, R_{2s-1} \in A$}\enspace,\\
    \multicolumn{2}{l}{g(R_{-2s}+a, R_{-2s+2}+a, \ldots, R_{-2}+a, a, R_{2}+a, \ldots, R_{2s}+a) - a}\: \\
    \multicolumn{2}{r@{}}{\text{if $R_{-2s+1} = R_{-2s+3} = \cdots = R_{2s-1} = a < 0$ and $-a \in A$}\enspace.}
    \end{array}\right.
\end{equation}
The first case is for the markers (and state $0$) which remain
unchanged, the second case is the simulation of $g$ in the even
piles. As proved in~\cite{CFM07}, for any $y \in A^\Z$ it holds
that $\xi(G(y)) = F(\xi(y))$. The images by $f$ of the remaining
ranges will be defined later on, first a few new notions need to
be introduced.

A sequence of consecutive piles $(x_i, \ldots, x_j)$ from a
configuration $x \in \B$ is said to be \emph{valid} if it is part
of an encoding of a CA configuration, \ie $x_i = x_{i+2} = \cdots
= x_j$ (these piles are markers) and for all $k \in \N$ such that
$0 \leq k < (j-i)/2$, $x_{i+2k+1} - x_i \in A$ (this is a valid
state). We extend this definition to configurations, when
$i=-\infty$ and $j=+\infty$, \ie $x \in \rho^c \circ \xi(A^\Z)$
for a given $c \in \Z$ ($x \in \B$ is valid if it is the raised
image of a CA configuration). A sequence (or a configuration) in
\emph{invalid} if it is not valid.

First we show that starting from a valid configuration, the SA \A is
nilpotent if and only if $\S$ is nilpotent. This is due to the fact
that we chose to put the markers at height~$0$, hence for any valid
encoding of the CA $x = \rho^c \circ \xi(y)$, with $y \in A^\Z$ and $c
\in \Z$,
\[
\lim_{n\to\infty} \dist_T(G^n(y), \cst{0}) = 0 \quad\text{if and
only if}\quad \lim_{n\to\infty} \dist(F^n(x), \cst{c}) =
0\enspace.
\]

It remains to prove that for any invalid configuration, \A is also
nilpotent. In order to have this behavior, we add to the local
rule $f$ the rules of the nilpotent automaton \SN for every
invalid neighborhood of width $4s+1$. For all ranges $R \in
\mathcal{R}^1_r$ not considered in Equation~\eqref{eq:fvalid},
\begin{equation}\label{eq:finvalid}
f(R) = \left\{\begin{array}{r@{\hspace{5mm}}l}
  -1 & \text{if $R_{-r} < 0$ or $R_{-r+1} < 0$ or $\cdots$ or $R_{r} < 0$}\enspace,\\
  0 & \text{otherwise.}
  \end{array}\right.
\end{equation}

Let $x \in \B$ be an invalid configuration. Let $k \in \Z$ be any
index such that $\forall l \in \Z$, $x_l \geq x_k$. Let $i,j \in
\Z$ be respectively the lowest and greatest indices such that $i
\leq k \leq j$ and $(x_i, \ldots, x_j)$ is valid ($i$ may equal
$j$). Remark that for all $n \in \N$, $(F^n(x)_i, \ldots,
F^n(x)_j)$ remains valid. Indeed, the markers are by construction
the lowest piles and Equations~\eqref{eq:fvalid}
and~\eqref{eq:finvalid} do not modify them. The piles coding for
non-zero states can change their state by
Equation~\eqref{eq:fvalid}, or decrease it by~$1$ by
Equation~\eqref{eq:finvalid}, which in both cases is a valid
encoding. Moreover, the piles $x_{i-1}$ and $x_{j+1}$ will reach a
valid value after a finite number of steps: as long as they are
invalid, they decrease by~$1$ until they reach a value which codes
for a valid state. Hence, by induction, for any indices $a, b \in
\Z$, there exists $N_{a,b}$ such that for all $n \geq N_{a, b}$
the sequence $(F^n(x)_a, \ldots, F^n(x)_b)$ is valid.

In particular, after $N_{-2Nr-1, 2Nr+1}$ step, there is a valid
sequence of length $4Nr+3$ centered on the origin (here, $N$ is
the number of steps needed by \S to reach the
configuration~$\cst{0}$, given by
Definition~\ref{def:canil}). Hence, after $N_{-2Nr,
2Nr}+N$ steps, the local rule of the CA \S applied on this valid
sequence leads to 3 consecutive zeros at positions $-1, 0, 1$. All
these steps are illustrated on Figure~\ref{fig:validation}.
\begin{figure}[!ht]
  \begin{center}
  \includegraphics[scale=0.7]{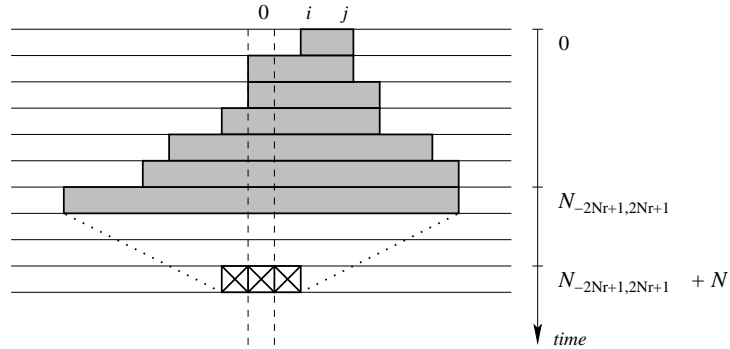}
  \caption{Destruction of the invalid parts. The lowest valid sequence (in gray) extends until it is large enough. Then after $N$ other steps the 3 central piles (hatched) are destroyed because the rule of the CA is applied correctly.}
  \label{fig:validation}
  \end{center}
\end{figure}

Similarly, we prove that for all $n \geq N_{-2Nr-k, 2Nr+k} + N$, the
sequence $(F^n(x)_{-k}, \ldots, F^n(x)_k)$ is a constant sequence
which does not evolve. Therefore, there exists $c \in \Z$ such that
$\lim_{n\to\infty} \dist(F^n(x), \cst{c}) = 0$.  We just proved that
\A is nilpotent, \ie $\lim_{n\to\infty} \dist(F^n(x), \cst{c}) = 0$
for all $x \in \B$, if and only if \S is nilpotent (because of the
equivalence of definitions given by
Corollary~\ref{cor:spreading-nil-equiv}), so \textbf{Nil} is
undecidable (Proposition~\ref{thm:spreading-nil-indec}).
\end{proof}
%---
\section{Conclusion}
In this article we have continued the study of sand automata, by
introducing a compact topology on the SA. In this new context of
study, the characterization of SA functions of \cite{CF03,CFM07} still
holds. Moreover, a topological conjugacy of any SA with a suitable CA
acting on a particular subshift might facilitate future studies about
dynamical and topological properties of SA, as for the proof of the
equivalence between equicontinuity and ultimate periodicity
(Proposition~\ref{prop:ult-equiv}).

Then, we have given a definition of nilpotency. Although it differs
from the standard one for CA, it captures the intuitive idea that a
nilpotent automaton ``destroys'' configurations. Even though nilpotent
SA may not completely destroy the initial configuration, they flatten
them progressively. Finally, we have proved that SA nilpotency is
undecidable (Theorem~\ref{th:nil-undec}). This fact enhances the idea
that the behavior of a SA is hard to predict. We also think that this
result might be used as a fundamental undecidability result, which
could be reduced to other SA properties.

Among these, deciding dynamical behaviors remains a major problem.
Moreover, the study of global properties such as injectivity and
surjectivity and their corresponding dimension-dependent decidability
problems could help understand if $d$-dimensional SA look more like
$d$-dimensional or $d+1$-dimensional CA. Still in that idea is the open
problem of the dichotomy between sensitive SA and those with
equicontinuous configurations. A potential counter-example would give
a more precise idea of the dynamical behaviors represented by SA.

\bibliographystyle{plain}
\bibliography{nilpsa}
%----
\end{document}